\documentclass{ecai} 

\usepackage{latexsym}
\usepackage{amssymb}
\usepackage{amsmath}
\usepackage{amsthm}
\usepackage{booktabs}
\usepackage{enumitem}
\usepackage{graphicx}
\usepackage{color}

\usepackage{algorithm}
\usepackage{amssymb}
\usepackage{algorithmicx}
\usepackage{algpseudocode}
\usepackage{calligra}
\usepackage{enumitem}
\usepackage{graphicx}
\usepackage{mathtools}
\usepackage{multicol}
\usepackage{multirow}
\usepackage{soul}
\usepackage{svg} 
\usepackage[colorinlistoftodos]{todonotes}
\usepackage{xcolor}

\newtheorem{theorem}{Theorem}
\newtheorem{lemma}[theorem]{Lemma}

\newtheorem{proposition}[theorem]{Proposition}

\newtheorem{definition}{Definition}

\newcommand{\FIXED}{\mathrm{Fixed}}
\newcommand{\M}{\mathcal{M}}

\newcommand{\PWin}{P}
\newcommand{\R}{\mathbb{R}}
\newcommand{\RV}{\mathsf{RV}}
\newcommand{\STAT}{\mathrm{Stat}}
\newcommand{\SECPRICE}{\mathrm{SecPrice}}
\newcommand{\SW}{\mathsf{SW}}

\newcommand{\vect}{\mathbf}

\newcommand{\BibTeX}{B\kern-.05em{\sc i\kern-.025em b}\kern-.08em\TeX}

\begin{document}


\begin{frontmatter}


\paperid{123} 


\title{Meta-mechanisms for Combinatorial Auctions over Social Networks}




\author[A]{\fnms{Yuan}~\snm{Fang}}
\author[A]{\fnms{Mengxiao}~\snm{Zhang}\thanks{Corresponding author. Email: mengxiao.zhang@uestc.edu.cn}}
\author[B]{\fnms{Jiamou}~\snm{Liu}} 
\author[A]{\fnms{Bakh}~\snm{Khoussainov}}

\address[A]{School of Computer Science and Engineering, University of Electronic Science and Technology of China, China}
\address[B]{School of Computer Science, The University of Auckland, New Zealand}


\begin{abstract}
Recently there has been a large amount of research designing mechanisms for auction scenarios where the bidders are connected in a social network. Different from the existing studies in this field that focus on specific auction scenarios e.g. single-unit auction and multi-unit auction, this paper considers the following question: is it possible to design a scheme that, given a classical auction scenario and a mechanism $\tilde{\M}$ suited for it, produces a mechanism in the network setting that preserves the key properties of $\tilde{\M}$? To answer this question, we design meta-mechanisms that provide a uniform way of transforming mechanisms from classical models to mechanisms over networks and prove that the desirable properties are preserved by our meta-mechanisms. Our meta-mechanisms provide solutions to combinatorial auction scenarios in the network setting:  (1) combinatorial auction with single-minded buyers and (2) combinatorial auction with general monotone valuation. To the best of our knowledge, this is the first work that designs combinatorial auctions over a social network. 
\end{abstract}

\end{frontmatter}


\section{Introduction}

Recently there has been a large amount of research extending various mechanisms in the classical auction scenarios, where the bidders have no connection with each other, to the settings where the bidders are connected in a social network. The motivation of this research comes from scenarios when the cost of traditional advertising for auction is large or when the seller aims to create a more organic and trusted form of marketing through word-of-mouth. These scenarios include, e.g., online marketplaces such as eBay or Etsy that encourage sellers to share their listings with their social circles \cite{tadelis2016reputation},  crowdfunding campaigns such as Kickstarter or Indiegogo  which rely on social networks to spread the word and attract backers \cite{silva2020success}, viral marketing campaigns, group buying, influencer marketing, and any other business models that benefit from network effects. 

One theme in this research is to discover whether key properties of mechanisms in classical auctions can be preserved in auctions where bidders are connected in a social network. Examples of classical auction scenarios include single-item auctions, multi-unit auctions, and various combinatorial auctions. Most of these auctions possess mechanisms that satisfy important properties of mechanism design, e.g., incentive compatibility (IC), individual rationality (IR) and non-deficiency (ND). Extending these properties of mechanisms to the setting of social networks turned out to be fruitful and  yet a challenging research direction \cite{guo2021emerging}. Most of these extensions are non-uniform as they are scenario dependent. Moreover, each scenario requires development of new ideas and techniques suited for the scenario, where techniques used for one scenario are not necessarily applicable to another scenario. E.g., \cite{li2017mechanism} focuses on single-unit auctions while \cite{fang2023multi} on multi-unit auctions. 

In this paper, we pose the following question: Is it possible to design a scheme that, given a classical auction scenario and a mechanism $\tilde{\M}$ suited for it, produces a mechanism in the social network setting that preserves the key properties of $\tilde{\M}$?  We call such schemes {\bf meta-mechanisms} as they are applied to mechanisms in classical settings and produce mechanisms where the bidders are connected via social networks. Thus, different from the previous studies that focus on specific auction scenarios, this paper aims to design meta-mechanisms that provide a uniform way of transforming mechanisms from classical models to mechanisms over networks. 

Such a uniform method serves two purposes: First, it will consolidate various existing auction approaches, enabling a more cohesive understanding across different scenarios. Second and more significantly, the meta-mechanism may facilitate the creation of new mechanisms for previously unexplored auction scenarios in networks, notably in combinatorial auctions. Such auctions, which involve bidders interested in item bundles with complex valuations, present unique challenges that are amplified by social network dynamics. 


\paragraph{Contributions.} We introduce MetaMSN, a novel meta-mechanism designed to transform classical auction mechanisms into formats suitable for auctions over social networks. A key feature of MetaMSN is its ability to inherit essential properties such as IC, IR, and ND from the classical mechanisms under the {\em non-sensitivity assumption}. Building upon MetaMSN, we initiate the study of combinatorial auctions over social networks. (1) MetaMSN gives rise to a combinatorial auction mechanism with {\em single-minded buyers} in a network setting, which also upholds the IC, IR, and ND properties. (2) We then modify MetaMSN to create MetaMSN-m, which maintains the IC, IR, and ND properties without any extra assumption. MetaMSN-m is then used to produce a combinatorial auction with general monotone valuation over social networks. Lastly, we provide a empirical comparative analysis of the social welfare achieved by our mechanisms for combinatorial auctions with single-minded, sub-modular and sub-additive buyer. These experiments underscore the practical effectiveness and potential of our proposed mechanisms in diverse auction contexts within social networks.

\paragraph{Related work. } There has been a growing interest in designing auctions over social networks. The IDM mechanism \cite{li2017mechanism} that focuses on single-unit auctions is IC, IR and ND. The key idea of IDM is to incentivise information propagation by rewarding those who hold critical positions in the social network. Mechanisms CMD \cite{li2019diffusion}, FDM \cite{zhang2020incentivize}, and NRM \cite{zhang2020redistribution} are also designed for single-unit auctions, with improved revenue, fairness, and social welfare, respectively.
More recently, several research attempt to extend to multi-unit auction scenarios. 
For this, a generalised IDM \cite{zhao2018selling} and DNA-MU \cite{kawasaki2020strategy} are proposed. However, both fail to preserve the IC property. 
Later, LDM-Tree \cite{liu2023diffusion} and MUDAN \cite{fang2023multi} are proposed using new ideas for incentivising information propagation. Both mechanisms apply the idea of confining the competition within a small group of buyers, thereby giving buyers incentives to  propagate the information. 
It turns out that our meta-mechanisms unify some of these mechanisms; See Sec.~\ref{sec:MetaMSN}. 
Several studies extend to auction scenarios with other assumptions, e.g., auctions with budgeted buyers \cite{xiao2022multi} and auctions with intermediaries \cite{li2018customer,li2020information,li2019diffusion,li2024diffusion}. 

Designing combinatorial auctions amounts to one of the key topics in mechanism design. VCG mechanism is the most famous combinatorial auction that ensures IC, IR and social welfare maximisation. However, VCG mechanism requires to find the winner(s) who maximise the social welfare, which is NP-hard \cite{dobzinski2005approximation,assadi2019improved}. 
Many mechanisms are proposed to circumvent the computational complexity issue. \cite{dobzinski2006truthful} proposes the first IC and computationally efficient mechanism for general monotone valuations. 
Subsequent studies propose mechanisms with improved approximation ratios w.r.t. social welfare for difference cases of valuation functions.  \cite{dobzinski2005approximation,dobzinski2007two,assadi2021improved} investigate the case with sub-additive valuation functions, while \cite{dobzinski2006truthful,dobzinski2007two,krysta2012online,dobzinski2016breaking,assadi2020improved} investigate the case with sub-modular valuation functions.
\cite{lehmann2002truth,mu2008truthful,archer2004approximate} focus on the case with low communication cost, i.e., combinatorial auction with single-minded buyers.
With our meta-mechanisms, these mechanisms would be conveniently extended to auctions over social networks. 
To the best of our knowledge, our work is the first on combinatorial auctions over social networks.

\section{Preliminaries}

In this section, we introduce the formal model of {\em auctions over social networks}. The classical auction model (i.e., without a social network) can be regarded as a special case. The definitions below are largely taken from earlier work in the field \cite{li2017mechanism,guo2021emerging}. 

An auction consists of a seller, denoted by $s$, and $n$ buyers $B=\{1,2,\ldots,n\}$. Seller $s$ has a set of $m$ indivisible items for sale. 
In a combinatorial auction, buyers are interested in purchasing {\em bundles} of items, which are represented by indicator vector $x\in \{0,1\}^m$. The empty bundle is $\emptyset\coloneqq \vect{0}$ and the bundle of all items is $x^0\coloneqq \vect{1}$.  Let $X$ denote the set of all possible bundles. A buyer $i\in B$ has a valuation to any bundle, as defined by the {\em valuation function} $v_i\colon X\to \R^+$. Namely, $v_i(x)$ denotes the amount that $i$ is willing to pay for the bundle $x\in X$. Following normal convention, we require that the valuation function $v_i$ is {\em normalised}, i.e., $v_i(\vect{0}) = 0$, and {\em monotone}, i.e., $v_i(x) \leq v_i(y)$, for all $x\subseteq y \subseteq X$.


We assume that the seller and the buyers form a social network, represented by a {\em directed graph} $G=(N,E)$, where the vertex set is $N \coloneqq \{s\}\cup B$, and edge set $E\subseteq N^2\backslash \{(i,i)\mid i\in N\}$ represents the social connections between the agents. 
Any edge $(i,j)$ denotes a (uni-directional) channel of information propagation from buyer $i$ to buyer $j$. 
The {\em neighbour} of each agent $i\in N$ is $r_i \coloneqq \{j\in N \mid (i,j)\in E\}$. In particular, $r_s$ is the neighbour set of seller $s$. 

We assume that the auction information is not publicly known by all vertices of the social network. At the very beginning, only the seller $s$ and her neighbours $r_s$ know the auction. 
The seller would request her neighbours to invite their neighbours to the auction, in the hope of attracting more buyers. Any buyer who joins the auction would decide whether to pass the information onto their neighbours, and so on. The auction information will then propagate along edges of the social network until no new buyer joins the auction. 

More specifically, we use $\theta_i\coloneqq (v_i,r_i)$ to denote the {\em (true) profile} of each buyer $i\in B$. The profile $\theta_i$ is initially hidden information known only by $i$. Upon joining the auction, the buyer $i$ would report a profile $\theta'_i=(v_i',r_i')$ to the seller $s$, where the reported valuation $v_i'\colon X\to \R^+$ and the reported neighbourhood $r_i'\subseteq N$. The reported profile $\theta'_i$ is not necessarily the same as the true profile $\theta'$, but rather, the buyer $i$ could strategically decide what information is passed to $s$ depending on how much benefit doing so would bring.  Following standard convention \cite{li2017mechanism}, we assume that the reported neighbour set $r_i'\subseteq r_i$ for any $i\in B$. 

Inductively, we say that a buyer $i\in B$ {\em joins} the auction if either $i\in r_s$, or $i\in r_j$ where $j$ is a buyer who joins the auction. Note that under this setup, a buyer $i$ can join the auction only when there is a {\em path} from $s$ to $i$. Clearly, a buyer $i$ who does not join the auction  does not report any profile to the seller, yet in this case we write  $\theta'_i=(0^X, \emptyset)$, treating the reported valuation all as 0.
The {\em global profile} is the reported profiles $\theta'\coloneqq \left(\theta'_j\right)_{i\in B}$ of all buyers. In particular, the {\em true global profile} is $\theta\coloneqq \left(\theta_i\right)_{i\in B}$.

Given a reported global profile, a mechanism determines how the seller $s$ allocates items among buyers and how much the buyers would pay to $s$. Next we give the formal definition of a mechanism and its desirable properties. Let $\Theta$ denote the set of all possible reported profiles. 

\begin{definition}
A {\em mechanism} $\M$ consists of two functions $(\pi,p)$, where $\pi\colon \Theta^n \to X^n$ is the {\em allocation function} and $p\colon \Theta^n \to \R^{n}$ is the {\em payment function}.
For a reported global valuation $\theta'$, the {\em allocation result} $\pi(\theta')$ is written as $\pi(\theta')\coloneqq (x_1,\ldots,x_n)$ and the {\em payment result} $p(\theta')$ as $(p_1,\ldots,p_n)$.    
\end{definition}
Intuitively, for a buyer $i\in B$, the results $\pi_i(\theta')=x_i$ and $p_i(\theta')$ show she wins the bundle $x_i$ by paying $p_i$.  We require that any allocation result is {\em feasible}, i.e., $\sum_{i\in B} x_i \leq \vect{1}$, i.e., each item is allocated at most once. 
In this paper, we will occasionally refer to a mechanism defined above as a {\em mechanism over social networks}, to distinguish it from mechanisms in the classical model, which we will discuss below.

Let $\M=(\pi,p)$ be a mechanism. 
Given the global profile $\theta'$, the {\em utility} of buyer $i$ is defined as $u_i(\theta')\coloneqq v_i(x_i)-p_i$. Intuitively, a rational buyer $i$ would report a profile $\theta'_i$ that leads to a high utility $u_i(\theta')$. 
More precisely, let $\theta_{-i}\coloneqq (\theta_1,\ldots,\theta_{i-1},\theta_{i+1},\ldots,\theta_n)$ denote the true profiles of all buyers but $i$. An ideal mechanism $\M$ should satisfy the following properties: 
\begin{enumerate}[leftmargin=*]
\item {\em Incentive compatibility (IC)}: $\M$ is IC if for any buyer, reporting truthfully is a dominant strategy: for all $i\in B$, $\theta_i,\theta_i'\in \Theta$ and $\theta_{-i}\in \Theta^{n-1}$,  $u_i(\theta_i, \theta_{-i})\geq u_i(\theta_i',\theta_{-i})$. 
\item {\em Individual rationality (IR)}: $\M$ is IR if any buyer by reporting truthfully receives non-negative utility: for all $i\in B$, $\theta_i\in \Theta$ and $\theta_{-i}\in \Theta^{n-1}$, $u_i(\theta_i,\theta_{-i})\geq 0$. 

\item {\em Non-deficiency (ND)}: The {\em revenue} $\RV(\theta')$ is the sum of the payment of all buyers, i.e., $\sum_{i=1}^n  p_{i}$. $\M$ is  ND if for any global profile $\theta'$,  $\RV(\theta')\geq 0$.
\end{enumerate}

Note that, IC ensures that the buyers would truthfully reveal both their valuations and their neighbourhoods. 
Aside from the properties above, one would also wish the mechanism to achieve a high social welfare. Namely, the {\em social welfare} $\SW(\theta')$ under the mechanism $\M$ is the sum of the utilities of all agents, i.e., $\SW(\theta')\coloneqq \sum_{i=1}^n v_i(x_i)$. 

\paragraph{The classical model.} When the seller $s$ is  connected with  all the buyers via edges, i.e., $E=\{s\}\times B$ and thus $r_s=B$, only the valuation of the buyers are hidden information. There is then no need for information propagation as all buyers in $B$ are assumed to join the auction by default, and the profile of each buyer $i$ is reduced to only the valuation, i.e., $\theta_i=v_i$. This case coincides with the classical auction model without taking into account the social ties between buyers. In this case, a mechanism $\M$ only needs to take the reported valuations $v'_i$, $i\in B$, as input, while the IC, IR, and ND properties of $\M$ can be defined only using the valuations of the buyers. 

Instead of crafting mechanisms for specific auction scenarios, we seek a generic scheme  that transforms a given mechanism designed for the classical model into a mechanism in the more general setup where the underlying social network is not assumed to be a star graph. We call such a generic scheme a {\em meta-mechanism}.

\section{Meta-Mechanism for Social Networks}
\label{sec:MetaMSN}

\subsection{Method} 
In this section, we present our meta-mechanism for social network, {\em MetaMSN}.  
For any buyer who are $h$-hops away from the seller, they can only join the auction when  some of the 1-hop neighbour of $s$ passes the information along their outgoing edges. However, 
as the heightened level of competition resulted from attracting more buyers to the auction potentially harm the utility of an existing buyer, a proper incentive mechanism is needed to ensure that buyers' reporting their neighbourhood truthfully. The key issue in designing an auction over social networks is thus to inject incentives to the mechanism so that allowing the neighbours $r_i$ of a buyer $i$ to join the auction {\em will not} harm the utility of $i$.

To mitigate this issue, we design a meta-mechanism based on the idea of {\em graph exploration}. In short, information propagation during the auction can be viewed as an exploration process of the graph $G(\theta')$ starting from the seller $s$. The {\em explored region} contains all buyers who have joined the auction. At the beginning of the auction, the seller $s$ ``explores'' only her neighbours in $r_s$. In each iteration, the seller would incentivise some buyers in the explored region to truthfully report their neighbourhoods so that more buyers would join the auction, thereby expanding the explored region. There are two ways in which a buyer would be incentivised: {\em exhausting} or {\em satisfying}. The former refers to declaring that this buyer would lose the bid and no longer be considered in the auction. The latter refers to satisfying the buyer's demand by securing their demanded items for them. In either case, allowing the buyer's neighbours to join the auction would not affect the buyer's own utility. At each iteration, a priority order is used to select a winner. In this way, the mechanism iteratively explores the graph $G(\theta')$ until the explored region cannot be expanded further.

We now define some necessary terminologies before formally presenting the meta-mechanism:
\begin{itemize}[leftmargin=*]
\item {\bf Residual item vector}:  Let the vector $x$ keep track of the items that are to be sold at the auction. At a given iteration, the vector is updated by  $x= x^0- \sum_{i\in B} x_i$, where $x_i$ is the vector indicating the items allocated to $i$.

\item {\bf Explored buyers $A$}: Initially, the set of explored buyers $A=r_s$, i.e., neighbours of $s$. Then at each iteration, the set $A$ is updated by declaring an explored buyer {\em exhausted} or a {\em winner}  (introduced below) using the following procedures:
\begin{itemize}[leftmargin=*]
    \item Repeatedly adding reported neighbours of exhausted agents until no more buyer can be added.
    \item Adding the reported neighbours of a chosen winner.
\end{itemize}

\item {\bf Potential winner set $\PWin$}: A {\em winner} refers to any buyer $i$ who has been allocated any item by the mechanism, i.e., $x_i\geq \vect{0}$. At the given iteration, a buyer $i\in A$ is a {\em potential winner} if there is a chance that $i$ wins in the current iteration. Specifically, a buyer $i$ is a potential winner if $i$ is a winner in $\tilde{\M}$ among $A\backslash W$.

\item {\bf Exhausted agent}: At the given iteration, a buyer in $A\backslash \PWin$ is called an {\em exhausted agents} as this buyer cannot win in the current iteration or any future iteration. The mechanism 
will ensure that an exhausted agent stays exhausted.

\item {\bf Priority $\sigma_i$}: Once the potential winner set $\PWin$ is updated, the algorithm sets {\em priority scores} $\sigma_i$ as the number of reported neighbours of $i, \forall i\in \PWin$,  as a metric to select a winner. 

\item {\bf Winner set $W$}: In each iteration, the buyer with the highest priority in $\PWin$ is selected as the winner and is added to $W$.


\item {\bf Termination condition}: Terminate the graph exploration if all explored buyers are either winners or exhausted, or all items are allocated, i.e., $P\backslash W = \emptyset$ or $x=\vect{0}$.  
\end{itemize}

As described in Algorithm~\ref{alg:MetaMSN}, the algorithm operates in an iterative manner, commencing with the initialization of several key components: the residual item vector $x$, the set of explored buyers $A$, the potential winner set $\PWin$, and the winner set $W$. The process continues until a point where either all buyers have been designated as winners or exhausted, or when all items have been allocated. During each iteration, MetaMSN methodically expands the explored buyer set. This expansion involves identifying and marking unexplored buyers, subsequently updating $A$ and $\PWin$. 
Upon completing the exploration of the graph and establishing the potential winner set, MetaMSN then assigns a priority score $\sigma_i$ to each potential winner $i\in \PWin$. The buyer with the highest priority score, denoted as $i^*$, is then selected for the current iteration. The algorithm determines the allocation for this buyer $i^*$ using the allocation function $\tilde{\pi}_{i^*}(\theta'_{A\backslash W})$, which is derived from the classical mechanism $\tilde{\M}$. Simultaneously, the payment of $i^*$ is set according to $\tilde{p}_{i^*}(\theta'_{A\backslash W})$, which is the payment function returned by $\tilde{\M}$. This process iteratively identifies and rewards winners until the termination condition is met.

\begin{algorithm}[h] 
\caption{The MetaMSN Algorithm} 
\label{alg:MetaMSN} 
\footnotesize
\begin{algorithmic}[1]
\Require Global profile $\theta'$ and classical $\tilde{\M}=(\tilde{\pi},\tilde{p})$
\Ensure Allocation result $\pi(\theta')$ and payment result $p(\theta')$
\State Initialise  $x \gets x^0$, $A \gets \{s\}$,  $W\gets \emptyset$ 
\State Initialise $\PWin$ as the winner set obtained by applying $\tilde{\M}$ on $\theta_{r_s}'$
\While{$\PWin\backslash W \neq \emptyset$ and $x \neq \vect{0}$}
    \While{$A$ contains an unmarked  $i\in W \cup (A\backslash \PWin)$} \label{ln:exploration_start}
         \State Update $A\gets A\cup r_i'$, mark  $i$
         \State Update $\PWin$ as the winners of $\tilde{\M}$ over $A\backslash W$
    \EndWhile \label{ln:exploration_end}
    \State Assign a priority $\sigma_i$ to each $i\in \PWin$ \label{ln:priority}
    \State Set $\pi_{i^*}(\theta'_{A\backslash W})\gets \tilde{\pi}_{i^*}(\theta'_{A\backslash W})$ and $p_{i^*}(\theta'_{ A \backslash W}) \gets \tilde{p}_{i^*}(\theta'_{A \backslash W})$ for $i^*\in \PWin$ with the top $\sigma_i$-priority \label{ln:pi&p}
    \State Update $W\gets W \cup \{i^*\}$ and  $x \gets x^0 -\sum_{i\in B}x_i$  
\EndWhile
\end{algorithmic}
\end{algorithm}

\vspace{-0.6cm}

\subsection{Analysis}

We now focus on the properties preserved by MetaMSN. 
For this, we make the following important definition: 


\begin{definition}    Let $\tilde{\M}$ be a mechanism in the classical auction model. We say that $\tilde{\M}$ has the {\em non-sensitivity property} if whenever a winner $w$ changes her reported valuation function $v_w'$, either $w$'s allocation becomes $0$ (i.e., $w$ ceases to be a winner), or the allocations of all buyers remain the same. 
\end{definition}

The non-sensitivity property captures an important class of mechanisms in the classical model. Below 
we give two simple examples of mechanisms that meet this property. 


\noindent{\bf (1) Single-unit auction.} Here, the seller $s$ has only one item for sell, i.e., $m=1$, which is valuated by each buyer $i \in B$ with a single value $v_i\in \R^+$. The second-price auction, i.e., one that allocates the item to the buyer with the highest reported valuation and charges the winner by the second highest reported valuation, 
achieves the IC, IR, and ND properties for single-unit auction in the classical setting. The second-price auction trivially satisfies the non-sensitivity property.



\noindent{\bf (2) Multi-unit auction.} Here, the seller $s$ has $m > 1$ homogeneous items for sell.  Each buyer $i\in B$ has an unit demand (i.e., has only positive valuation to a single item) and reports $v_i\in \R^+$ to express her valuation of a single item. The $(m+1)$-th price auction (MPA) generalises the second-price auction by allocating an item to each of the buyers who report the $m$-highest valuations and charges them by the $(m+1)$-highest reported valuation. If a winner misreports her valuation, then either she would no longer be a winner, or her reported valuation stays within the $m$-highest among all reported valuation, thereby not affecting the allocation of other buyers. Thus MPA satisfies the non-sensitivity property. 


\begin{theorem}\label{thm:MetaMSN}
Suppose $\tilde{\M}$ is a mechanism for an auction scenario that satisfies the non-sensitivity property. Then the mechanism $\M$ obtained from $\tilde{\M}$ by applying MetaMSN is IC, IR, \& ND, whenever $\tilde{\M}$ is IC, IR, \& ND, respectively. 
\end{theorem}

\begin{proof}
For {\bf IR}, consider the mechanism $\M$ obtained by applying MetaMSN to $\tilde{\M}$. If a buyer $i$ is exhausted by $\M$, then the utility of $i$ is 0. Otherwise,  $i$ is selected as a winner in some iteration. By definition, $\M$ maintains the allocation and payment rules of $\tilde{\M}$ for $i$, thereby ensuring $u_i(\theta_{A\backslash W})=\tilde{u}_i(\theta_{A\backslash W})$. The IR property thus directly follows from that of $\tilde{\M}$. 


\smallskip

For {\bf IC}, we first show that {\em no buyer can benefit from misreporting her valuation}. During an iteration of $\M$, a buyer $i\in A\backslash W$ falls in one of the following three cases:

\noindent{\bf Case 1:} 
$i\in \PWin$ is selected as a winner in this iteration when she reports $\theta_i'=(v_i,r_i')$ with her true valuation $v_i$. Write $\theta_{-i}$ for the profiles of buyers in $A\backslash W$ except $i$. Line~\ref{ln:pi&p} of Algorithm~\ref{alg:MetaMSN} sets $i$'s payment and allocation as $\pi_{i}((v_i,r_i'), \theta_{-i})= \tilde{\pi}_{i}(v_i, \theta_{-i})$ and $p_{i}((v_i,r_i'), \theta_{-i})= \tilde{p}_{i}(v_i, \theta_{-i})$, resp. By IC  of $\tilde{\M}$,  $u_i((v_i,r_i'), \theta_{-i}) = \tilde{u}_i(v_i,\theta_{-i}) \geq \tilde{u}_i(v_i',\theta_{-i}) = u_i((v_i',r_i), \theta_{-i})$, for any $\theta_{i}, \theta_{i}'=(v_i',r_i')$, and $\theta_{-i}$.

\noindent{\bf Case 2:} $i \in \PWin$ is not selected as a winner in this iteration when she reports profile $\theta_i'=(v_i,r_i')$. By definition of potential winner, $i$'s allocation in $\tilde{\M}$ is $\tilde{\pi}_{i}(v_i,\theta_{-i})\neq \vect{0}$, and her allocation in $\M$ is $\pi_{i}((v_i,r_i'),\theta_{-i})= \vect{0}$ and payment is $0$. There are two sub-cases: {\bf (a)} If $i$ misreports a profile $\theta_i'=(v_i',r_i')$ such that her allocation $\tilde{\pi}_{i}(v_i,\theta_{-i})$ in $\tilde{\M}$ is still not $\vect{0}$, then by non-sensitivity of $\tilde{\M}$, all other buyers' allocations remain the same and so do the potential winner set $\PWin$ and $i$'s priority.  Thus her allocation $\pi_{i}(\theta'_{A\backslash W})$ is still $\vect{0}$ and her utility is $u_i((v_i',r_i'), \theta_{-i}) = u_i((v_i,r_i'), \theta_{-i})=0$. {\bf (b)} If she misreports a valuation $\theta_i'=(v_i',r_i')$ such that her allocation  $\tilde{\pi}_{i}(v_i,\theta_{-i})$ in $\tilde{\M}$ becomes $\vect{0}$, she is exhausted in $\M$. Hence $i$'s utility $u_i((v_i',r_i'), \theta_{-i}) =0 \leq u_i((v_i,r_i'), \theta_{-i})$.


\noindent{\bf Case 3:} $i\notin \PWin$ when she reports profile $(v_i,r_i')$. Again, by IC of $\tilde{\M}$, we have $u_i((v_i,r_i'), \theta_{-i}) = \tilde{u}_i(v_i',\theta_{-i}) \geq \tilde{u}_i(v_i',\theta_{-i}) = u_i((v_i',r_i'), \theta_{-i}),$ for all $\theta_{i}'=(v_i',r_i')$. 

It remains to show that {\em no buyer can benefit from misreporting her neighbourhood}. Consider a buyer $i\in A\backslash W$ in an iteration. 

\noindent{\bf Case 1.} $i$ is the selected winner  when she reports her true neighbourhood $r_i$. Imagine $i$ hides some of her neighbours, i.e., reporting $r_i'\subset r_i$.  Then if she remains the top priority, $i$'s allocation and payment would not change, and her utility $u_i((v_i',r_i'),\theta_{-i})=u_i((v_i',r_i),\theta_{-i})$. If, on the other hand, her priority is no longer the highest, she will then not be a winner in this iteration and her utility drops, i.e., $u_i((v_i',r_i'),\theta_{-i})=0\leq u_i((v_i',r_i),\theta_{-i})$. 

\noindent{\bf Case 2.} $i$ is not a selected winner when she reports her true neighbourhood $r_i$. Hiding any of her neighbour would not increase her priority. Hence, she is not allocated any item and $u_i((v_i',r_i'),\theta_{-i})=u_i((v_i',r_i),\theta_{-i})=0$. 

Lastly, {\bf ND} of $\M$ trivially follows ND of $\tilde{\M}$.
\end{proof}

By Theorem~\ref{thm:MetaMSN}, both the second-price auction and MPA can be transformed into auctions over social networks that are IC, IR, and ND. We further remark that in the multi-unit auction scenario, the resulting mechanism transformed from MPA coincides with the MUDAN mechanism \cite{fang2023multi}. Our work demonstrates that MUDAN belongs to a general class of auctions in a social network that are produced by MetaMSN.  

There is no mechanism that is IC, IR, ND, and maximises social welfare \cite{takanashi2019efficiency}. However, we have the lower bound of the social welfare achieved by MetaMSN under certain assumptions. 

\begin{theorem}\label{thm:MetaMSN_sw}
Given a classical mechanism $\tilde{\M}$ that maximises social welfare, the social welfare of $\M$ obtained from $\tilde{\M}$ by applying MetaMSN is no less than the social welfare of $\tilde{\M}$ over the seller's neighbourhood.
\end{theorem}

\begin{proof}

Let $w^{(t)}$ be the selected winner, $W^{(t)}$ be the winner set, $A^{(t)}$ be the explored buyer set, and $x^{(t)}$ be the residual item vector at the end of the $t$-th iteration of MetaMSN $\M$. Note that, for any $t$, we have $x^{(t)} = x^{(t-1)} - x_{w^{(t)}}$,  $W^{(t)} = W^{(t-1)} \cup \{w^{(t)}\}$, and $A^{(t)} \supseteq A^{(t-1)}$. 
Now suppose we run the mechanism $\tilde{\M}$ over the buyers in $A^{(t)} \backslash W^{(t-1)}$ and items in $x^{(t)}$. Denote by $\SW_{\tilde{\M}}^{(t)}(\cdot)$ the resulting social welfare (i.e., sum of utilities of buyers in $A^{(t)} \backslash W^{(t-1)}$ and the seller). Note that the input to $\SW_{\tilde{\M}}^{(t)}(\cdot)$ is the input to the mechanism $\tilde{\M}$. 
Let $\theta^{(t)}_{i}$ be the gained utility of buyer $i$, i.e., her utility plus her payment, in such an $\tilde{\M}$.
Then, $\SW_{\tilde{\M}}^{(t)}(\cdot)$ can be written as the gained utility of a winner, say $w^{(t)}$,  and the total utility of other buyers. That is,  $\SW_{\tilde{\M}}^{(t-1)} \left(\theta'_{A^{(t)}\backslash W^{(t-1)}}\right) = \theta^{(t)}_{w^{(t)}} + \SW_{\tilde{\M}}^{(t)} \left(\theta'_{A^{(t)}\backslash W^{(t)}} \right)$. 
This is equivalent to 
$$\SW_{\tilde{\M}}^{(t)} \left(\theta'_{A^{(t)}\backslash W^{(t)}} \right) = \SW_{\tilde{\M}}^{(t-1)} \left(\theta'_{A^{(t)}\backslash W^{(t-1)}}\right) -  \theta^{(t)}_{w^{(t)}}. $$

Given the equation above, we have 
\begin{equation}
\label{eqn:induction}
\begin{aligned}
& \SW_{\tilde{\M}}^{(t)} \left(\theta'_{A^{(t+1)}\backslash W^{(t)}}  \right) + \sum_{1 \leq l \leq t } \theta^{(l)}_{w^{(l)}}
\\
\geq &  \SW_{\tilde{\M}}^{(t)} \left(\theta'_{A^{(t)}\backslash W^{(t)}}  \right) + \sum_{1 \leq l \leq t } \theta^{(l)}_{w^{(l)}}
\\
 = & \SW_{\tilde{\M}}^{(t-1)} \left(\theta'_{A^{(t)}\backslash W^{(t-1)}}\right) -  \theta^{(t)}_{w^{(t)}} + \sum_{1 \leq l \leq t } \theta^{(l)}_{w^{(l)}}\\
= & \SW_{\tilde{\M}}^{(t-1)} \left(\theta'_{A^{(t)}\backslash W^{(t-1)}}\right) + \sum_{1 \leq l \leq t -1 } \theta^{(l)}_{w^{(l)}}.  \\        
\end{aligned}
\end{equation}

\noindent The first inequation holds as the classical mechanism $\tilde{\M}$ maximises the social welfare, and the social welfare of $\tilde{\M}$ over a larger buyer set $A^{(t+1)}\backslash W^{(t)}$ must be no less than that over a smaller one $A^{(t)} \backslash W^{(t)}$. 

Now we consider the implementation of the MetaMSN mechanism $\M$. Assume that $\M$ terminates in the $T$-th iteration. We have 
\begin{align*}
    \SW_{\M}(\theta') &= \SW_{\tilde{\M}}^{(T-1)} \left(\theta'_{A^{(T)}\backslash W^{(T-1)}}  \right) + \sum_{1 \leq l \leq T -1 } \theta^{(l)}_{w^{(l)}} \\
    &\geq \SW_{\tilde{\M}}^{(T-2)} \left(\theta'_{A^{(T-1)}\backslash W^{(T-2)}}  \right) + \sum_{1 \leq l \leq T -2 } \theta^{(l)}_{w^{(l)}} \\
    &\geq \ldots  \\
    &\geq \SW_{\tilde{\M}}^{(1)} \left(\theta'_{A^{(2)}\backslash W^{(1)}}  \right) +  \theta^{(1)}_{w^{(1)}} \\
    &\geq \SW_{\tilde{\M}}^{(0)} \left(\theta'_{A^{(1)}\backslash W^{(0)}}  \right)\\ 
    &\geq \SW_{\tilde{\M}}^{(0)} \left(\theta'_{r_s \backslash W^{(0)}}  \right) \\
    &=\SW_{\tilde{\M}} (\theta'_{r_s}).
\end{align*}
The first and the last equations are derived due to Line~\ref{ln:pi&p} of Alg.~\ref{alg:MetaMSN} while the inequations are due to Equation~\eqref{eqn:induction}. 
\end{proof}

\section{Combinatorial Auction with Single-minded Buyers}
To study combinatorial auction design over social networks, we first focus on the special case when the buyers are {\em single-minded}, i.e., each buyer demands a specific and publicly-known bundle $x_i\in X$ and submits a single bid to express her valuation $v_i$ for $x_i$.   In this scenario, no partial allocation is allowed, i.e., each buyer is allocated either the entire bundle she demands, or nothing. The valuation of any other bundle $x'\neq x_i$ is 0.  This scenario  provides a simple  yet practically-relevant framework which serves as a stepping stone towards understanding more general combinatorial auction scenarios.

For this scenario, we revisit the  mechanism proposed by Lehmann, Ocallaghan, and Shoham (LOS) \cite{lehmann2002truth}. For each buyer $i\in B$, define the {\em average valuation}, denoted by $av_i$,  as the valuation $v_i$ for the desired bundle $x_i$ divided by the number of items in the bundle, i.e., $av_i\coloneqq v_i / (x_i \cdot \vect{1})$. 
The LOS mechanism employs a strategy where buyers are ranked in a descending order of their average valuations. Items are then allocated sequentially to these buyers according to this ranking. If any item in a buyer's desired bundle has been allocated to another buyer, the mechanism bypasses this buyer and moves on to the next in line. This process continues until either all items are allocated, or every buyer has been considered. The payment required from each allocated buyer $i$ is determined by the average valuation of the next buyer in the sorting sequence, multiplied by the number of items in $i$'s bundle, i.e., $p_i=av_{i+1} (x_i \cdot \vect{1})$. 



\begin{lemma}
    The LOS mechanism is non-sensitive. 
\end{lemma}

\begin{proof}
    Consider an arbitrary winner $w$ who misreports her valuation. If $w$ is still a winner, the mechanism would allocate the same bundle to her. Moreover, the average valuations of other buyers do not change and so does their relative ordering. Thus, 
the allocation of the others would be the same. 
\end{proof}

The following thus naturally follows from Theorem~\ref{thm:MetaMSN}.

\begin{theorem}
Let LOS-SN denote the mechanism over social networks resulted from applying MetaMSN to the LOS mechanism. 
    LOS-SN satisfies IC, IR, and ND. \qed
\end{theorem}

\section{Combinatorial Auction with General Valuation Buyers}

\subsection{The DNS Mechanism: A Recap}
Our next goal is to design combinatorial auctions over social networks in the general setting. 
It is widely acknowledged that the deterministic mechanism for combinatorial auctions, such as VCG mechanism, is not computational efficient \cite{assadi2019improved}.  
We therefore resort to {\em randomised auction mechanisms} which define a probability distribution over a number of deterministic mechanisms. For any property Q of a deterministic mechanism, such a randomised auction is said to satisfy Q {\em in the universal sense} (Q-U) if any deterministic mechanism as its outcome satisfies the Q property. 

We now briefly recall the randomised mechanism proposed by Dobzinski, Nisan and Schapira  (DNS) \cite{dobzinski2006truthful}. The DNS mechanism operates as follows:

\begin{enumerate}[leftmargin=*]
    \item Buyer Grouping: Buyers are randomly allocated into three distinct groups: $\SECPRICE$, $\FIXED$, and $\STAT$. The allocation probabilities are $1-\epsilon$ for $\SECPRICE$, $\epsilon/2$ for $\FIXED$, and $\epsilon/2$ for $\STAT$, where $0 < \epsilon < 1$ is a parameter controlling the approximation ratio of social welfare.
    
    \item Optimal Fractional Solution: The mechanism computes the optimal fractional solution, denoted as $opt_{\text{stat}}$, among the buyers in the $\STAT$ group.
    
    \item Second-Price Auction for the Grand Bundle: A second-price auction is conducted for the grand bundle $x^0 = \vect{1}$ among buyers in the $\SECPRICE$ group. The reserve price is set as $p_{\text{sec}} = \frac{opt_{\text{stat}}}{\sqrt{m}}$. If there is a winning buyer, all items are allocated to her; otherwise, proceed to the next step.
    
    \item Fixed-Price Auction: In this stage, a fixed-price auction is held among the $\FIXED$ group buyers. Buyers are sorted arbitrarily, and items are allocated iteratively. During each iteration, a price vector $p_{\text{fix}} \in \mathbb{R}^m$ is set, where unsold items are priced at $\frac{\epsilon \cdot opt_{\text{stat}}}{8m}$ and sold items at infinity. Each buyer $i$ is allocated a bundle $x_i$ that maximises their utility, i.e., $x_i \in \arg\max_{y \in X} \left\{ v_i(y) - p_{\text{fix}} \cdot y \right\}$, followed by an update to the set of unsold items.
\end{enumerate}

Further insights into this mechanism can be found in \cite{dobzinski2006truthful}. The DNS mechanism offers a structured approach to allocating items and setting prices in a randomised combinatorial auction, aiming to achieve a balance between social welfare approximation and the IC and IR constraints. This mechanism satisfies IC-U, IR-U, and ND-U.

The DNS mechanism does not satisfy  non-sensitivity, and indeed, as we see in the next proposition, applying MetaMSN to DNS fails to preserve incentive compatibility. 

\begin{proposition}
The DNS-SN mechanism, derived from applying MetaMSN to the DNS mechanism, is not IC-U. 
\end{proposition}
\begin{proof}
For simplicity, we focus on the fixed-price auction part of the DNS mechanism. Our objective is to show that buyers can manipulate their advantage by misrepresenting their valuations in this segment of the auction. The argument is illustrated through an example.

Consider a social network as depicted in Fig.~\ref{fig:GraphDNS}, where the seller $s$ has two items $u_1, u_2$. The valuations of buyers $a, b, c$ are listed in Fig.~\ref{tab:valuations}. Assume $a, b, c$ are consistently assigned to $\FIXED$ in each iteration of the mechanism. We analyse the outcomes under truthful reporting in the DNS-SN iterations:

\noindent{\bf Iteration 1:} Buyers $a$ and $b$ are included in the explored set $A$. Assume that we arrive at Step 4 (Fixed-Price Auction) with buyers $a, b$, the price $p_{\text{fix}} = (2, 2)$,  buyer $a$ preceding $b$ in the ordering. As bundles $(1,0)$ and $(0,1)$ maximise the utilities for $a$ and $b$ resp, DNS allocates $u_1$ to $a$ for a price of \$2, and \( u_2 \) to \( b \) for \$2. The potential winner set is updated to \( \PWin = \{a, b\} \). Given \( b \)'s higher priority, \( b \) wins this iteration.

\noindent{\bf Iteration 2:} Assume we arrive at Step 4 with buyers $a,c$, $p_{\text{fix}}=(2, 2)$, and $c$ preceding $a$ in the ordering. DNS allocates $u_2$ to $c$ for \$2, making $c$ the winner of this iteration.
In this truthful scenario, buyer $a$ does not win any item, yielding a utility $0$.

However, the scenario changes if $a$ misrepresents her valuation. Consider a rerun of DNS-SN with $a$ reporting a valuation of $v_a(1,1) = 7$. In the first iteration, $b$'s neighbour $c$ is added to $A$. With the same fixed price as before  and $a$ preceding $c$ in the order, DNS allocates the bundle $(1,1)$ to $a$, as it maximises $a$'s utility under the falsified valuation. As the sole potential winner, $a$ wins the bundle $(1,1)$, leaving $b$ and $c$ with no items. In this misrepresentation scenario, $a$'s utility increases to \$1, higher than that obtained through truthful reporting. This example shows that the DNS-SN mechanism fails to uphold the IC-U property when subjected to strategic misreporting in the fixed-price auction. \end{proof}

\begin{figure}
\begin{minipage}{.23\textwidth}\centering
\includegraphics[width=0.55\columnwidth]{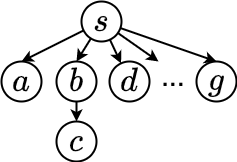}
\caption{A social network}
\vspace{0.5cm}
\label{fig:GraphDNS}
\end{minipage}
\begin{minipage}{.23\textwidth}\phantom{\rule{1cm}{0.2cm}}
\footnotesize
    \begin{tabular}{|c|ccc|}
    \hline
          &  $(1, 0)$&  $(0,1)$&  $(1,1)$\\ \hline
         $a$&  \$4&  \$0&  \$5\\
         $b$&  \$0&  \$3&  \$3\\
        $c$& \$5& \$0 & \$5\\
 \hline
    \end{tabular}
    \caption{Valuations of $a,b,c$}
    \vspace{0.5cm}
    \label{tab:valuations}
\end{minipage}

\end{figure}

\vspace{-0.4cm}

\subsection{MetaMSN-m}
In this section, we introduce {\em MetaMSN-m}, a modified version of our MetaMSN. The primary objective of MetaMSN-m is to retain incentive compatibility across transformations from classical IC mechanisms to their counterparts in network settings, without necessitating additional assumptions. 

The central innovation of MetaMSN-m lies in its approach to winner selection. Unlike its predecessor, MetaMSN, which restricts the selection to at most one winner per iteration, MetaMSN-m adopts a strategy aimed at maximising buyer demand in each iteration. It achieves this by utilising the underlying classical mechanism to potentially accommodate multiple buyers simultaneously.

The specifics of MetaMSN-m are encapsulated in Algorithm~\ref{alg:MetaMSN-m}. The key distinction between this algorithm and its predecessor (Algorithm~\ref{alg:MetaMSN}) is manifest in Line~\ref{ln:pi&p_m}. Here, MetaMSN-m diverges by determining allocations for all buyers in set $A\backslash W$ in a given iteration, based directly on the allocation rules stipulated by the classical mechanism $\tilde{\M}$. 


\begin{algorithm}[h] 
\caption{The MetaMSN-m Algorithm} 
\label{alg:MetaMSN-m} 
\begin{algorithmic}[1]
\footnotesize
\Require Global profile $\theta'$ and classical mechanism $\tilde{\M}=(\tilde{\pi},\tilde{p})$
\Ensure Allocation result $\pi(\theta')$ and payment result $p(\theta')$ 
\State Initialise $x \gets x^0, A \gets \{s\},$  $W\gets \emptyset$ 
\State Initialise $\PWin$ as the winner set obtained by applying $\tilde{\M}$ on $\theta_{r_s}'$
\While{$\PWin\backslash W \neq \emptyset$ and $x \neq \vect{0}$}
    \While{$A$ contains an unmarked  $i\in W \cup A\backslash \PWin$} \label{ln:exploration_start_m}
         \State Update $A\gets A\cup r_i'$, mark $i$
         \State Update $\PWin$ as the winners of $\tilde{\M}$ over $A\backslash W$
    \EndWhile \label{ln:exploration_end_m}
    \State Set $\pi_i(\theta'_{A \backslash W}) \gets \tilde{\pi}_i(v'_{A \backslash W})$ and $p_i(\theta'_{ A \backslash W}) \gets \tilde{p}_{i}(v'_{A \backslash W})$ for all $i \in A\backslash W$ \label{ln:pi&p_m} 
    \State Update $W\gets W \cup \PWin$ and $x \gets x^0 -\sum_{i\in B}x_i$ \label{ln:update_m} 
\EndWhile
\end{algorithmic}
\end{algorithm}

The next theorem demonstrates that MetaMSN-m inherits the IC, IR and ND properties from classical mechanisms. It can be proved by similar arguments for Thm.~\ref{thm:MetaMSN}. 

\begin{theorem}\label{thm:MetaMSN-m}
Suppose that a classical mechanism $\tilde{\M}$ is IC, IR \& ND. Then the corresponding mechanism $\M$ obtained by applying MetaMSN-m is also IC, IR \& ND. 
\end{theorem}

\begin{proof}
For {\bf IR} and {\bf ND}, we prove the properties of $\M$ using similar arguments for IR and ND in Theorem~\ref{thm:MetaMSN}. 

For {\bf IC}, we first show that {\em no buyer can benefit from misreporting her valuation}. During an iteration of $\M$, consider a buyer $i\in A\backslash W$. Write $\theta_{-i}$ for the profiles of buyers in $A\backslash W$ except $i$. Suppose $i$ reports a profile $\theta_i=(v_i,r_i')$. 
Line~\ref{ln:pi&p_m} of Algorithm~\ref{alg:MetaMSN-m} sets $i$'s payment and allocation as $\pi_{i}((v_i,r_i'), \theta_{-i})= \tilde{\pi}_{i}(v_i, \theta_{-i})$ and $p_{i}((v_i,r_i'), \theta_{-i})= \tilde{p}_{i}(v_i, \theta_{-i})$, resp. By IC  of $\tilde{\M}$,  $u_i((v_i,r_i'), \theta_{-i}) = \tilde{u}_i(v_i,\theta_{-i}) \geq \tilde{u}_i(v_i',\theta_{-i}) = u_i((v_i',r_i), \theta_{-i})$, for any $\theta_{i}, \theta_{i}'=(v_i',r_i')$, and $\theta_{-i}$.  

Next we can show that {\em no buyer can benefit from misreporting her neighbourhood} using similar arguments in Thm.~\ref{thm:MetaMSN}.  
\end{proof}

As the DNS mechanism is IC, IR and ND, the next theorem then follows directly from Theorem~\ref{thm:MetaMSN-m}. 

\begin{theorem}
The mechanism over social networks derived from applying the MetaMSN-m mechanism to the DNS mechanism satisfies IC, IR and ND in the universal sense.  \qed
\end{theorem}

The following theorem shows the lower bound of social welfare of MetaMSN-m; See the proof in Appendix~\ref{app:meta-MSN-m}. 

\begin{theorem}\label{thm:MetaMSN-m_sw}
Given a classical mechanism $\tilde{\M}$ that maximises social welfare, the social welfare of $\M$ obtained from $\tilde{\M}$ by applying MetaMSN-m is no less than the social welfare of $\tilde{\M}$ over the seller's neighbourhood. \qed
\end{theorem}

\paragraph{Limitation of MetaMSN-m.} While MetaMSN-m exhibits certain advantages over its predecessor MetaMSN, it is important to acknowledge a notable limitation in terms of social welfare optimisation. The primary issue with MetaMSN-m lies in its allocation strategy, which tends to allocate items as soon as there is a demand within an iteration. To illustrate, consider a scenario where the aggregate demand from buyers within the seller's immediate neighbours $r_s$ surpasses the total supply of items. In such a case, the MetaMSN-m mechanism  may allocate all items exclusively to these neighboring buyers in $r_s$ without propagating information to further levels. This allocation method precludes the mechanism from exploring and leveraging potential higher valuations that could arise from a broader competitive environment, which may result in a lower social welfare. Nevertheless, we will demonstrate in the next section using empirical validations that the resulting mechanism over social network may still provide sufficiently high social welfare.

\section{Experiment}

We conduct empirical studies to verify the performance of our meta-mechanisms. All the experiments are performed in Python on macOS Monterey system with Apple M2 Pro CPU and 16GB RAM. 
The code is available on \url{https://anonymous.4open.science/r/msn-878B/}. 
Our experiments serve to answer the following questions: 
{\bf Q1.} Social welfare and revenue are two key performance metrics of a mechanism. How do our mechanisms over social networks perform in terms of these two criteria, when compared with the {\em optimal case}, i.e., when the corresponding classical mechanism is applied to the {\em entire population} in the social network without the need for information propagation? 
{\bf Q2.} Information propagation through the social network attracts new buyers to the auction. What is the effect of information propagation in terms of the number of  new buyers joining the auction, and how much social welfare and revenue would grow, when compared with the case where only buyers in $r_s$ join the auction? 
{\bf Q3.}  MetaMSN-m tends to produce a mechanism over social network that has an inferior social welfare than MetaMSN, when applied to the same classical mechanism. How inferior MetaMSN-m will get compared to MetaMSN? 
For each questions above, we focus on three important and well-explored combinatorial auction scenarios: (i) the case with single-minded buyers (for Q1-Q3) where both MetaMSN and MetaMSN-m are IC, and two other cases where MetaMSN-m is IC (for Q1, Q2), i.e., (ii) the case with {\em sub-modular buyers} ($\forall x, y \in X, v_i(x\cup y) + v_i(x\cap y) \leq v_i(x) + v_i(y)$), and (iii) the case with {\em sub-additive buyers} ($\forall x, y \in X, v_i(x+y) \leq v_i(x) + v_i(y)$). 

\paragraph{Dataset.} We use three real-world datasets as the social network of a seller and buyers, including
Facebook social network \cite{leskovec2012learning},
Hamsterster friendships \cite{Kunegis2013konect}, 
and email-Eu-core network \cite{Yin2017Local}. 
The key statistics of the datasets are listed in Tab. \ref{tab:dataset} of App.~\ref{app:exp}. 
We randomly select one vertex as the seller and treat the others as the buyers. 
As the initial setup, especially the neighbour set of the seller, may effect experiment results, we repeat each scenario $|V|/20$ times and calculate the revenue and social welfare as the result for the scenario. 

\paragraph{Valuation.} We generate four sets of synthetic valuations. 
(i) For combinatorial auction with single-minded buyers, we generate a random vector in $ \{0,1\}^m$ as the demanded bundle and a random scalar as the average valuation per item for each buyer. The average valuation of each buyer $i \in B$ is sampled from (a) $\mathcal{U}(0,200000)$ and (b) $\mathcal{N}(100000,4000)$.
(ii) For combinatorial auction with monotone valuation buyers, we generate $2^m-1$ random scalar numbers for each buyer, each of which is the valuation of a bundle $x\in \{0,1\}^m$. Specifically, we use {\em coverage function} \cite{chakrabarty2015recognizing} and {\em square root function} to generate sub-modular and sub-additive valuations, resp. 
(iii) For sub-modular valuations, we define a finite set $S=\{1,2,\ldots, 400000\}$. Given the set $X$ of all possible bundles and the finite set $S$, associate each bundle $x\in X$ with a subset $s_x \subset S$. In particular, each bundle $x$ of a single item is associated with a random subset $s_x \subset S$, whose size is a random number drawn from $\mathcal{U}(1, 200000)$. For each bundle $y\in X$, the associated subset is $s_y = \cup_{i \in y} s_i $. Then the valuation of bundle $x$ is defined as $v(x)=|s_x|$. 

\paragraph{Benchmark.} 
We compare our meta-mechanisms with (a) the classical mechanism applied to a social network with $s$ connecting with all buyers, denoted by ALL, and (b) the classical mechanism applied to the seller's neighbourhood, denoted by FIRST. These two cases give us the upper and lower bounds of $\SW$ and of $\RV$ that any mechanism extending classical mechanism to social networks could achieve, respectively. 

\begin{figure}[t]
    \centering
    \includegraphics[width = 0.8\columnwidth]{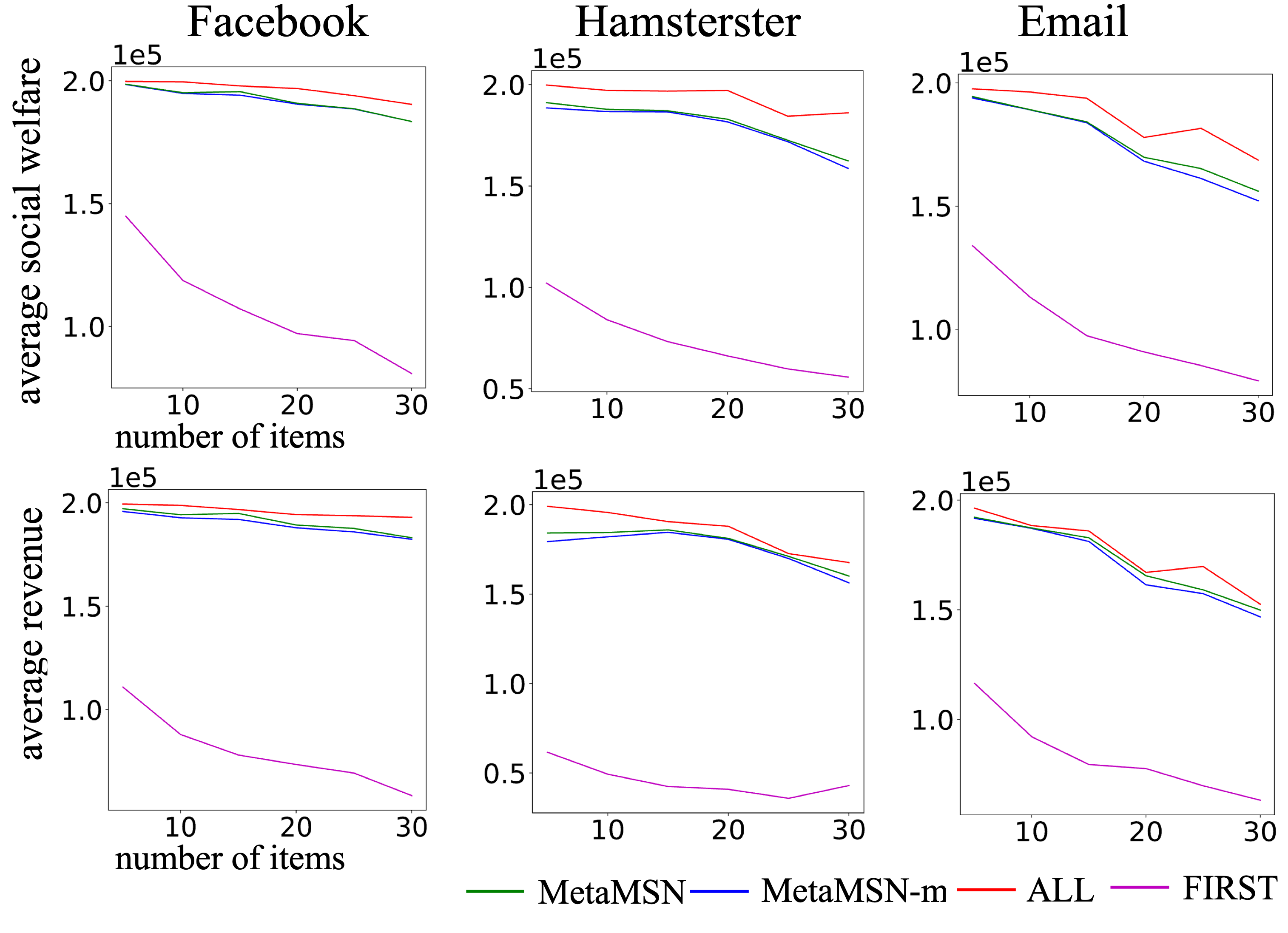}
    \caption{Social welfare and revenue of four mechanisms in three datasets for (i) combinatorial auction with single-minded buyers}
    \label{fig:singleminded}
    \vspace{0.5cm}
\end{figure}

\paragraph{Results.} Figure~\ref{fig:singleminded} shows the social welfare $\SW$ and the revenue $\RV$ per item with the increase of the number of items $m$ for Case (i) with single-minded buyers under uniformly distributed valuations. 
Specifically, for Q1, ALL obtains higher social welfare and revenue than MetaMSN and MetaMSN-m, under varying $m$. Nevertheless, the lines of MetaMSN are very close to those of ALL, losing by at most $3\%$ ($5\%$),  $9\%$ ($10\%$),  and $8\%$ ($7\%$) of $\SW$ ($\RV$) for Facebook, Hamsterster, and Email networks, resp., which shows that MetaMSN achieves near-optimal social welfare and revenue. 
For Q2, both MetaMSN and MetaMSN-m obtain better social welfare and revenue than FIRST. MetaMSN grows by at least $37\%$ ($79\%$),  $87\%$ ($201\%$),  and $46\%$ ($66\%$) of $\SW$ ($\RV$) while MetaMS-m grows by at least $36\%$ ($78\%$),  $85\%$ ($199\%$),  and $45\%$ ($63\%$) of $\SW$ ($\RV$) for the three datasets, resp. The growth verifies that effectiveness of our meta-mechanisms in information prorogation. For Q3, MetaMSN obtains better $\SW$ and $\RV$ than MetaMSN-m. Nevertheless, the gap is very small. That is because when selling a small bundle over a large network, MetaMSN-m actually can explore a lot through exhausting buyers in each iteration due to the high competition. The results under normally distributed valuations show similar patterns; See Figure~\ref{fig:singleminded-gause} in App.~\ref{app:exp}.

Table~\ref{tab:result_submodular} presents the $\SW$ and $\RV$ for Case (ii) with sub-modular buyers. The results for Case (ii) are similar to those for Case (i), i.e., ALL performs better than MetaMSN-m, followed by FIRST.  
For Q1, MetaMSN-m returns inferior social welfare and social welfare than ALL under varying $m$, losing by at most $7\%$ ($17\%$) of $\SW$ ($\RV$) across the three datasets. For Q2, MetaMS-m grows by at least $6\%$ ($6\%$) of $\SW$ ($\RV$) across the three datasets. 

\begin{table}[t]
    \centering
    \caption{Social welfare and revenue of three mechanisms in three datasets for (ii) combinatorial auction with sub-modular buyers}
    \vspace{0.5cm}
    \label{tab:result_submodular}
    \resizebox{\columnwidth}{!}{
    {\begin{tabular}{|c|c|c|c|c|c|c|c|}
        \hline
        \multirow{2}{*}{Dataset} & \multirow{2}{*}{$m$} & \multicolumn{3}{|c|}{Social welfare} & \multicolumn{3}{|c|}{Revenue} \\ \cline{3-8}
         &  & MetaMSN-m &{ALL} &FIRST & MetaMSN-m &{ALL} &FIRST\\\hline
        \multirow{3}{*}{Facebook} & 3 &30469 &31059 &30089 &21315 &23334 &19731\\ \cline{2-8}
         &4 &33222 &35633 &31432	&26795 &29195 &22422\\ \cline{2-8}
         & 5 &31629 &32061 &31531 &18516 &18833 &17470\\ \hline \hline
        \multirow{3}{*}{Hamsterter}  & 3 &31081 &33453 &24553 &26108 &31439 &13122\\ \cline{2-8}
         & 4 &33728 &36438	&27148	&28911	&33484 &13376\\ \cline{2-8}
         & 5 &33093 &35738 &27951 &27218 &32297 &14133\\ \hline\hline
         \multirow{3}{*}{email} & 3 &33683 &35195	&27957	&31609	&34575	&19743\\ \cline{2-8}
         & 4 &34008 &35249	&28526 &28089	&31416	&18256\\ \cline{2-8}
         & 5 &34007 &36523	&28916 &28462 &33186 &15930\\ \hline
    \end{tabular}}
    }
\end{table}


\paragraph{Limitation.} While the result demonstrates our meta-mechanisms' strengths in adapting auctions for social networks, we acknowledge certain limitations to the experiments that offer avenues for future research. Specifically, our experiments focus on sub-modular and sub-additive valuation buyers, are conducted on relatively small networks, and do not thoroughly explore the effects of the parameter $\epsilon$ in the DNS mechanism. Despite these limitations, the findings show near-optimal social welfare and revenue performance and provide a foundation for further refinement of auctions in social networks. For more detailed discussions, see App.~\ref{app:limit}.

\section{Conclusion}
In this paper, we design meta-mechanisms that provide a uniform way of transforming mechanisms from classical models to mechanisms over networks while preserving desirable properties. Our meta-mechanisms also provide the first solution to combinatorial auction over networks.



\begin{ack}
This work is partially supported by National Natural Science Foundation of China No. 62172077 and Research Fund for International Senior Scientists No. 62350710215. 
\end{ack}



\bibliography{main}

\clearpage


\appendix

\noindent {\Large  \bf APPENDIX}


\section{Meta-MSN-m}
\label{app:meta-MSN-m}


\noindent {\bf Theorem~\ref{thm:MetaMSN-m_sw}.}
{\it Given a classical mechanism $\tilde{\M}$ that maximises social welfare, the social welfare of $\M$ obtained from $\tilde{\M}$ by applying MetaMSN-m is no less than the social welfare of $\tilde{\M}$ over the seller's neighbourhood.}

\begin{proof}

Let $A^{(1)}$ be the set of explored buyers at the end of the first iteration, and we have $r_s \subseteq A^{(1)}$. 
As the classical mechanism $\tilde{\M}$ maximises the social welfare for all possible valuation sets, the social welfare of $\tilde{\M}$ over the explored buyers should be no less than that of $\tilde{\M}$ over seller's neighbours, i.e., $\SW_{\tilde{\M}}(\theta_{A^{(1)}}) \geq \SW_{\tilde{\M}}(\theta_{r_s})$. Also, by Line~\ref{ln:pi&p_m} of Alg.~\ref{alg:MetaMSN-m}, we have $\SW_{\M}(\theta'_{A^{(1)}}) =  \SW_{\tilde{\M}}(\theta'_{A^{(1)}})$, which is $\geq \SW_{\tilde{\M}}(\theta_{r_s})$. 

If the items are allocated in one iteration, $\SW_{\M}(\theta'_{A^{(1)}})$ is the social welfare of $\M$ over the buyer set $B$, and we have proved what we want. Otherwise, the mechanism continues to expand the set of explored buyers. With the same arguments, we can show that in every iteration $t$, the social welfare $\SW_{\M}(\theta'_{A^{(t)}}) \geq \SW_{\tilde{\M}}(\theta_{r_s})$, which proves the theorem.  
\end{proof}

\section{Experiment}
\label{app:exp}

Here, we give the complementary content of our experiment.

\noindent{\bf Dataset.} We first list the key statistics of the datasets in Table~\ref{tab:dataset}.

\begin{table}[h]
    \centering
    \footnotesize
    \caption{Dataset statistics. $C$ denotes clustering coefficient}
    \vspace{0.5cm}
    \begin{tabular}{|c|c|c|c|c|}
    \hline
    dataset & $|V|$ & $|E|$ & $C$ & diameter \\
    \hline
    Facebook social network & 4039 & 88234 & 0.6055 & 8 \\
    \hline
    Hamsterster friendships & 1858 & 12534 & 0.0904 & 14 \\
    \hline
    email-Eu-core network & 1005 & 25571 & 0.3994 & 7  \\
    \hline
    \end{tabular}
    \label{tab:dataset}
\end{table}

\paragraph{Results.} Table~\ref{tab:result_subadditive} presents the $\SW$ and $\RV$ for Case (iii) combinatorial auction with sub-additive buyers. The results for Case (iii) are also similar to those for Case (i), i.e., ALL performs better than MetaMSN-m, followed by FIRST.  For Q1, MetaMSN-m returns inferior social welfare and social welfare than ALL under varying $m$, losing by at most $17\%$ ($31\%$) of $\SW$ ($\RV$) across the three datasets. For Q2, MetaMS-m grows by at least $6\%$ ($5\%$) of $\SW$ ($\RV$) across the three datasets.

\begin{table}[h]
    \centering
    \caption{Social welfare and revenue for Case (iii) combinatorial auction with sub-additive buyers}
    \vspace{0.5cm}
    \label{tab:result_subadditive}
    \resizebox{\columnwidth}{!}{
    \begin{tabular}{|c|c|c|c|c|c|c|c|}
        \hline
        \multirow{2}{*}{Dataset} & \multirow{2}{*}{$m$} & \multicolumn{3}{|c|}{Social welfare} & \multicolumn{3}{|c|}{Revenue} \\ \cline{3-8}
         &  & MetaMSN-m &ALL &FIRST & MetaMSN-m &ALL &FIRST\\\hline
        \multirow{3}{*}{Facebook} 
        & 3 &34156 & 36156 & 31453 & 33058 & 34125 & 25825 \\ \cline{2-8}
        & 4 &34715 & 35438 & 31960 & 21026 & 30352 & 11503 \\ \cline{2-8}
        & 5 &42154 & 44160 & 38331 & 29656 & 36730 & 8434  \\ \hline\hline
        \multirow{3}{*}{Hamsterter}  
        & 3 &31417 & 33678 & 29010 & 26634 & 32841 & 14318 \\ \cline{2-8}
        & 4 &35092 & 38173 & 33000 & 26265 & 37922 & 25016 \\ \cline{2-8}
        & 5 &33513 & 40359 & 29780 & 22867 & 28656 & 15267 \\ \hline\hline
         \multirow{3}{*}{email} 
        & 3 &33312 & 33709 & 28922 & 27676 & 32153 & 25037 \\ \cline{2-8}
        & 4 &37162 & 38116 & 35235 & 36033 & 37353 & 33074 \\ \cline{2-8}
        & 5 &33831 & 37975 & 28861 & 21900 & 23534 & 14148 \\ \hline
    \end{tabular} }
\end{table}

Figure~\ref{fig:singleminded-gause} shows the experiment results on LOS-SN under normal distribution $\mathcal{N}(100000, 4000)$. The results show similar patterns to those observed under the uniform distribution. 
For Q1, ALL obtains higher social welfare and revenue than MetaMSN and MetaMSN-m, under varying $m$. Nevertheless, the lines of MetaMSN are very close to those of ALL, losing by at most $2\%$ ($3\%$),  $2\%$ ($4\%$),  and $3\%$ ($4\%$) of $\SW$ ($\RV$) for Facebook, Hamsterster, and Email networks, resp., which shows that MetaMSN achieves near-optimal social welfare and revenue. 
For Q2, both MetaMSN and MetaMSN-m obtain better social welfare and revenue than FIRST. MetaMSN grows by at least $56\%$ ($73\%$),  $97\%$ ($181\%$),  and $52\%$ ($58\%$) of $\SW$ ($\RV$) while MetaMS-m grows by at least $36\%$ ($54\%$),  $82\%$ ($149\%$),  and $47\%$ ($49\%$) of $\SW$ ($\RV$) for the three datasets, resp. The growth verifies that effectiveness of our meta-mechanisms in information prorogation. For Q3, MetaMSN obtains better $\SW$ and $\RV$ than MetaMSN-m. Nevertheless, the gap is very small.

\begin{figure}[t]
    \centering
    \includegraphics[width = \columnwidth]{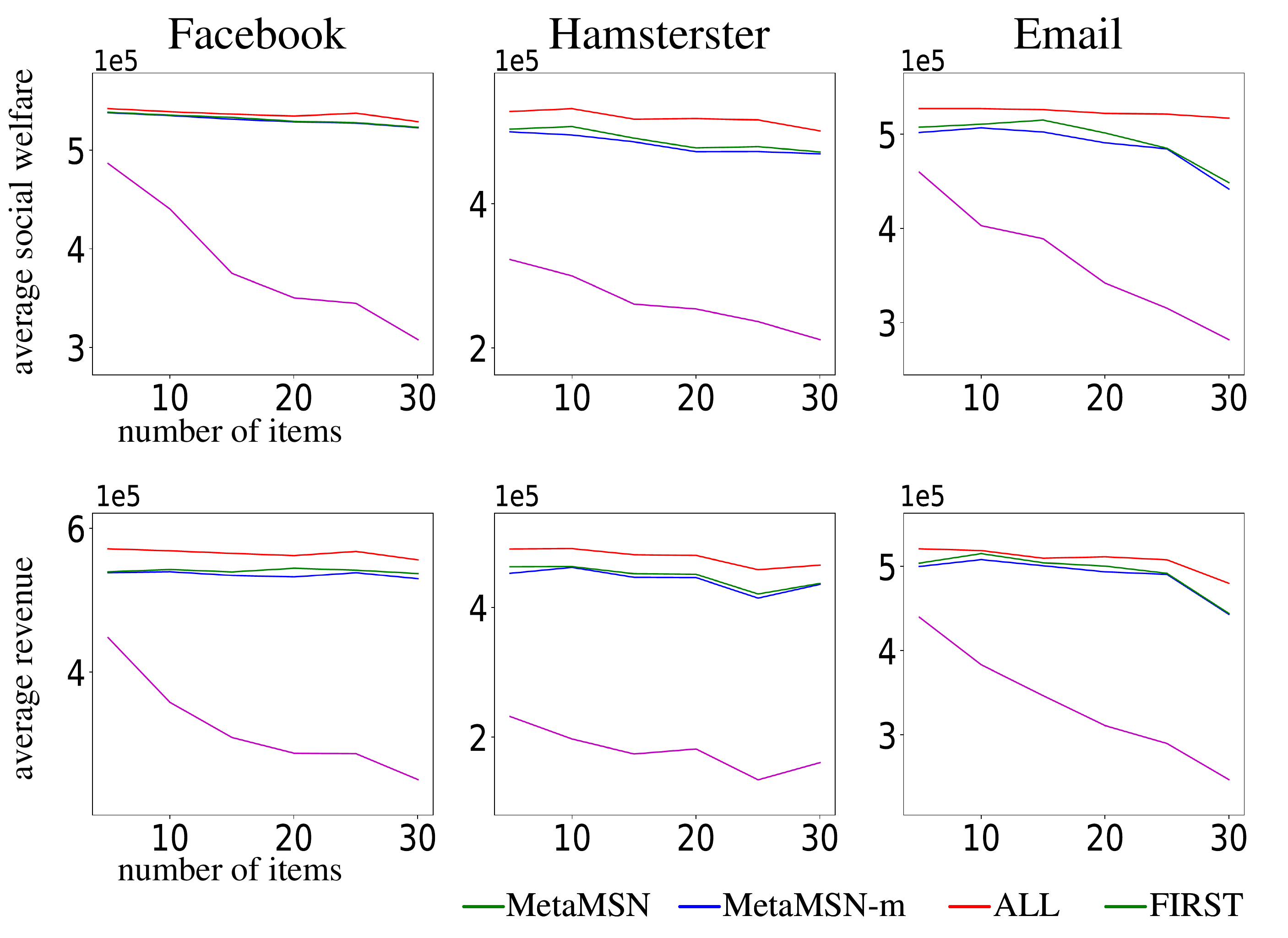}
    \caption{Social welfare and revenue of four mechanisms in three datasets for (i) combinatorial auction with single-minded buyers under normally distributed valuations}
    \label{fig:singleminded-gause}
    \vspace{0.6cm}
\end{figure}

\section{Limitations and Future Work of the Experiment}\label{app:limit}
While the study demonstrates considerable strengths in adapting auction mechanisms for social networks, we acknowledge certain limitations that offer avenues for future research. 
\begin{itemize}[leftmargin=*]
\item  Our experiments focus on sub-modular and sub-additive valuation buyers. While these two classes of buyers are important and widely-studied, our MetaMSN-m mechanism is indeed applicable to a much more general class of auctions, i.e., combinatorial auction with general mononotone valuation buyers. Furthermore, both the class of sub-additive functions and the class of sub-modular functions contain multiple specific functions and we test only those functions generated by the coverage function and square root function, respectively in each class. One of the future directions is to conduct more thorough experiments for other classes of combinatorial auction scenarios. 
 
\item  We apply our MetaMSN-m to the DNS mechanism by setting the parameter $\epsilon=0.01$. This is because the DNS mechanism empirically obtains a reasonable social welfare when $\epsilon=0.01$. When $\epsilon$ becomes too big, e.g., when $\epsilon=0.1$, the way DNS mechanism (in the classical model) works makes the social welfare drastically smaller than the case when $\epsilon=0.01$. When $\epsilon$ becomes too small, say when $\epsilon=0.005$, then too few buyers will participate in the fixed-price auction in the DNS mechanism. We must point out that the affect of this parameter on social welfare is not the focus of this study. Nevertheless, a future work is to could conduct more validation with different values of $\epsilon$.

\item  The experiments might seem limited in terms of the size and scale, possibly affecting the robustness and depth of the findings. Indeed, our experiments are conducted on a social network with at most $4000+$ nodes, and the number $m$ of items on sale is only $30$ for single-minded buyers and 5 for sub-modular and sub-additive valuation cases. This relatively small scale is due to high computational cost of the DNS mechanism which requires finding the optimal fractional solution of a group of buyers. Nevertheless, this does not hinder the role of our meta-mechanism which could transforms an arbitrary classical mechanism to the setting where bidders are connected in a social network.

\item The article acknowledges that there might be practical challenges in implementing the proposed mechanisms in real-world social networks. These challenges could stem from the complexities of real-world social dynamics, the need for robust and scalable technology solutions, and the necessity of aligning with legal and ethical standards in different jurisdictions. These are beyond the scope of this paper.

\end{itemize} 

Notably, the study was constrained by page limitations, precluding a complete and thorough empirical analysis. Despite these constraints, the research makes a substantial contribution by demonstrating the practicality of applying traditional auction principles in network settings, preserving key properties like incentive compatibility and individual rationality, and showing near-optimal social welfare and revenue performance. These findings lay the groundwork for further exploration and refinement of auction mechanisms in the context of social networks.


\end{document}